\documentclass[10pt,final,conference]{IEEEtran}

\usepackage{siunitx}       %
\usepackage{amsmath}    %
\usepackage{amssymb}                %
\usepackage{amsthm}                 %

  \usepackage{balance}

  \usepackage[font=footnotesize,hypcap=true]{caption}
  \usepackage[font=footnotesize,hypcap=true]{subcaption}
\usepackage{graphicx}
\usepackage[hyperref,x11names,cmyk,pdftex,table]{xcolor}
\usepackage{tikz}            %
\usetikzlibrary{arrows}

\usepackage{booktabs}
\usepackage{multirow}
\usepackage{dcolumn}
\usepackage{arydshln} %

\usepackage[]{algorithm2e}

\usepackage{enumerate}
\usepackage{paralist}

\usepackage[resetfonts]{cmap}
\usepackage[T1]{fontenc}
\usepackage[utf8]{inputenc}

  \usepackage[final,layout={inline},author=]{fixme}

  \usepackage[pdftex,breaklinks=true,bookmarks=false]{hyperref}
  \usepackage{cite}

\usepackage{cleveref}

\title{Fast Graphlet Transform of Sparse Graphs}

\makeatletter
\let\titletext\@title
\makeatother

  \author{%
    \IEEEauthorblockN{%
      Dimitris~Floros\IEEEauthorrefmark{1}                     \qquad
      Nikos~Pitsianis\IEEEauthorrefmark{1}\IEEEauthorrefmark{2} \qquad
      Xiaobai~Sun\IEEEauthorrefmark{2}}
    \\
    \IEEEauthorblockA{\small%
      \begin{tabular}{c @{\qquad\qquad} c}
        \IEEEauthorrefmark{1}%
        Department~of~Electrical~and~Computer~Engineering
        &
        \IEEEauthorrefmark{2}%
        Department~of~Computer~Science
        \\
        Aristotle~University~of~Thessaloniki
        &
        Duke~University
        \\
        Thessaloniki~54124,~Greece
        &
        Durham,~NC~27708,~USA
      \end{tabular}%
    }%
  }

  \date{}

\newcommand{\pdfauthors}{%
  D. Floros, N. Pitsianis, X. Sun}

\graphicspath{%
  {./figures/}%
}

\hypersetup{%
  pdffitwindow=false,%
  pdfstartview={FitH},%
  colorlinks,%
  pdfauthor={\pdfauthors},%
  pdftitle={\titletext},%
  citecolor=PaleGreen4,%
  filecolor=DarkOrchid4,%
  linkcolor=OrangeRed4,%
  urlcolor=black%
}

\makeatletter
\def\th@plain{%
  \thm@notefont{}%
  \itshape %
}
\def\th@definition{%
  \thm@notefont{}%
  \normalfont %
}
\makeatother

\newtheorem{theorem}{Theorem}

\newtheorem{lemma}{Lemma}

\setkeys{Gin}{draft=false}

\pdfpageattr{/Group <</S /Transparency /I true /CS /DeviceRGB>>}

\let\leftorig\left
\let\rightorig\right
\renewcommand{\left}{\mathopen{}\mathclose\bgroup\leftorig}
\renewcommand{\right}{\aftergroup\egroup\rightorig}

\crefname{equation}{}{}

\crefname{figure}{Fig.}{Fig.}
\Crefname{figure}{Fig.}{Fig.}

\FXRegisterAuthor{df}{edf}{\color{blue}[DF]}
\FXRegisterAuthor{tl}{etl}{\color{cyan}[TL]}
\FXRegisterAuthor{np}{enp}{\color{magenta}[NP]}
\FXRegisterAuthor{xs}{exs}{\color{olive}[XS]}

\hypersetup{final}

\DeclareSIUnit{\nothing}{\relax}

\begin{document}

  \bstctlcite{IEEEexample:BSTcontrol}

  \maketitle

\addcontentsline{toc}{section}{Abstract}
\begin{abstract}
We introduce the computational problem of graphlet transform of a
sparse graph. Graphlets are fundamental topology elements of all
graphs/networks. They can be used as coding elements to encode
graph-topological information at multiple granularity levels, for
classifying vertices on the same graph/network, as well as, for making
differentiation or connection across different networks.
Network/graph analysis using graphlets has growing applications. We
recognize the universality and increased encoding capacity in using
multiple graphlets, we address the arising computational complexity
issues, and we present a fast method for exact graphlet transform. The
fast graphlet transform establishes a few remarkable records at once
in high computational efficiency, low memory consumption, and ready
translation to high-performance program and implementation. It is
intended to enable and advance network/graph analysis with graphlets,
and to introduce the relatively new analysis apparatus to graph
theory, high-performance graph computation, and broader applications.

 \end{abstract}


\section{Introduction}
\label{sec:introduction}

\begin{figure}
  \centering
\newcommand{\figGraphletWidth}{0.2\linewidth}
\newcommand{\graphletID}{0}
\hspace*{\fill}
\renewcommand{\graphletID}{0}
\begin{subfigure}{\figGraphletWidth}
  \includegraphics[width=\linewidth]{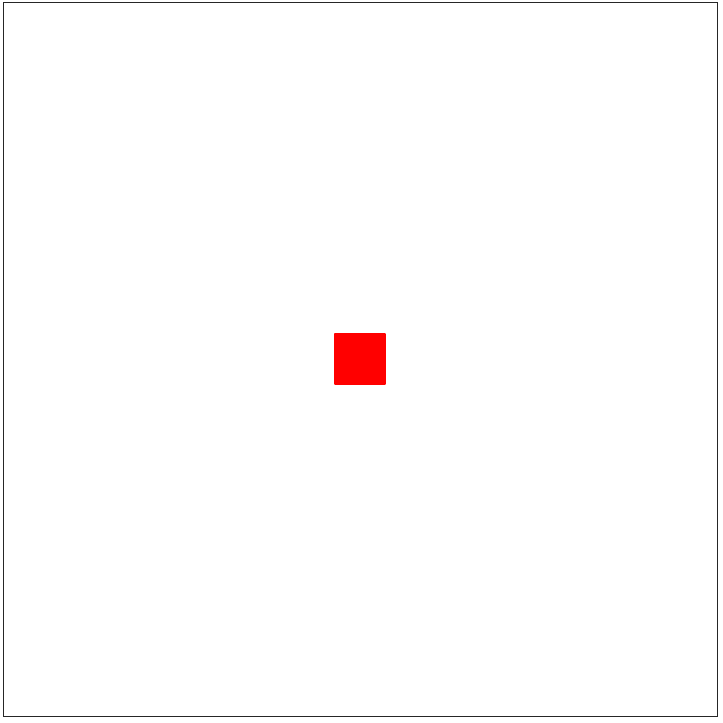}
  \caption*{$\boldsymbol\sigma_{\mathbf{\graphletID}}$}
\end{subfigure}
\hspace*{\fill}
 \renewcommand{\graphletID}{4}
\begin{subfigure}{\figGraphletWidth}
  \includegraphics[width=\linewidth]{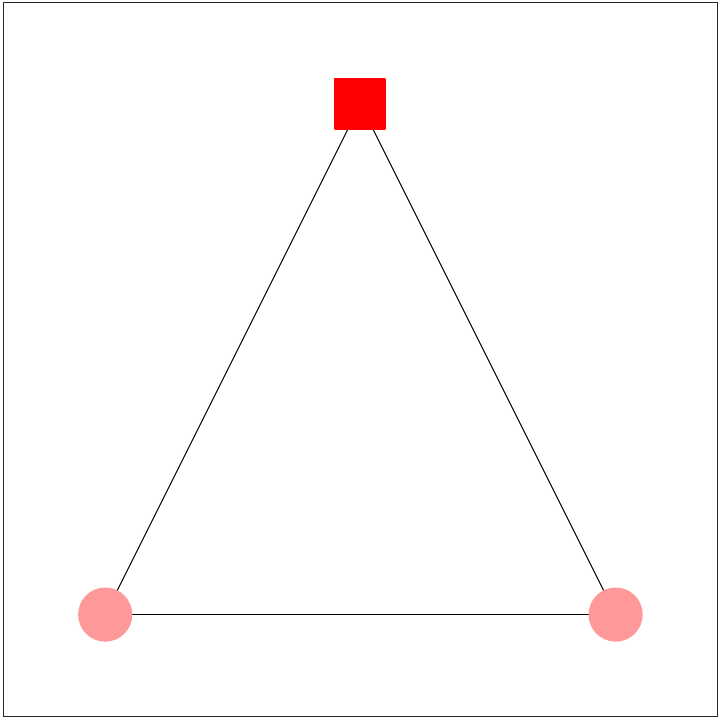}
  \caption*{$\boldsymbol\sigma_{\mathbf{\graphletID}}$}
\end{subfigure}
\hspace*{\fill}
 \renewcommand{\graphletID}{8}
\begin{subfigure}{\figGraphletWidth}
  \includegraphics[width=\linewidth]{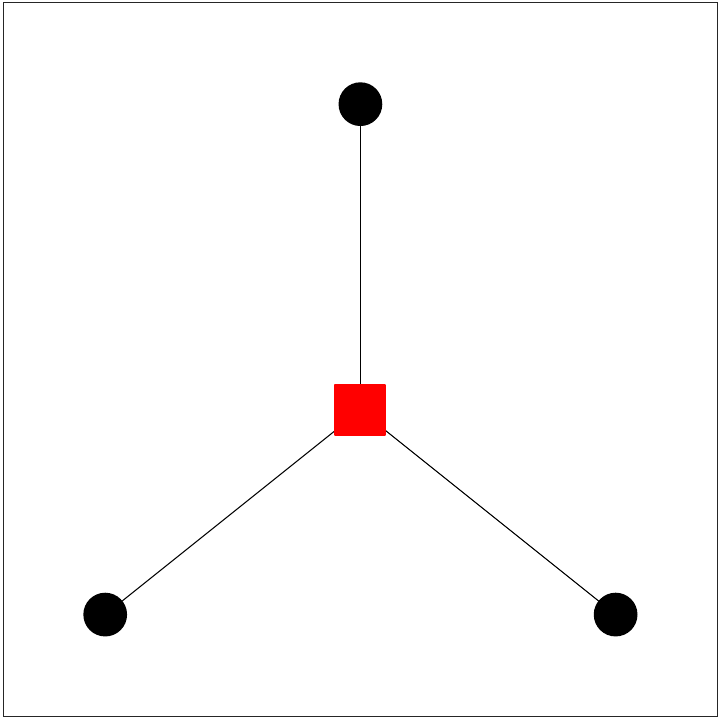}
  \caption*{$\boldsymbol\sigma_{\mathbf{\graphletID}}$}
\end{subfigure}
\hspace*{\fill}
 \renewcommand{\graphletID}{12}
\begin{subfigure}{\figGraphletWidth}
  \includegraphics[width=\linewidth]{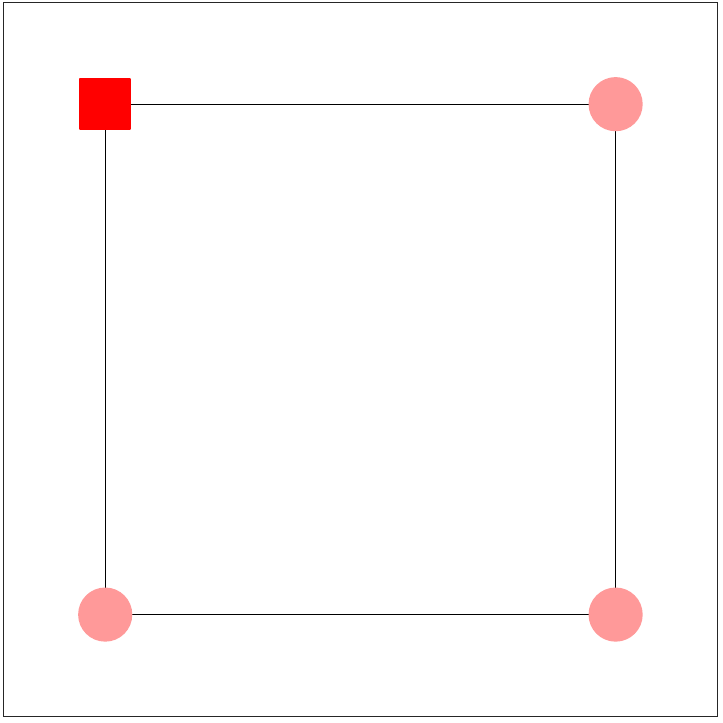}
  \caption*{$\boldsymbol\sigma_{\mathbf{\graphletID}}$}
\end{subfigure}
\hspace*{\fill}
 \\[0.5em]
\hspace*{\fill}
\renewcommand{\graphletID}{1}
\begin{subfigure}{\figGraphletWidth}
  \includegraphics[width=\linewidth]{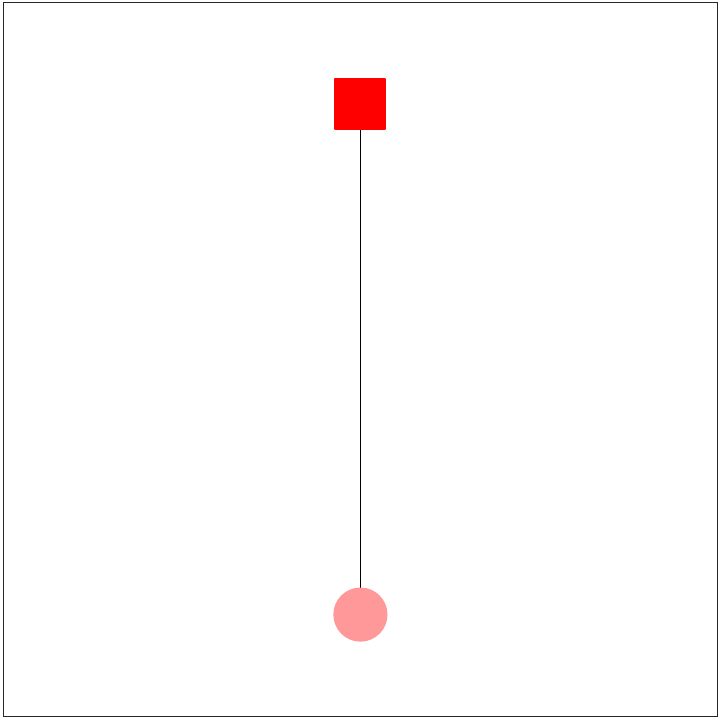}
  \caption*{$\boldsymbol\sigma_{\mathbf{\graphletID}}$}
\end{subfigure}
\hspace*{\fill}
 \renewcommand{\graphletID}{5}
\begin{subfigure}{\figGraphletWidth}
  \includegraphics[width=\linewidth]{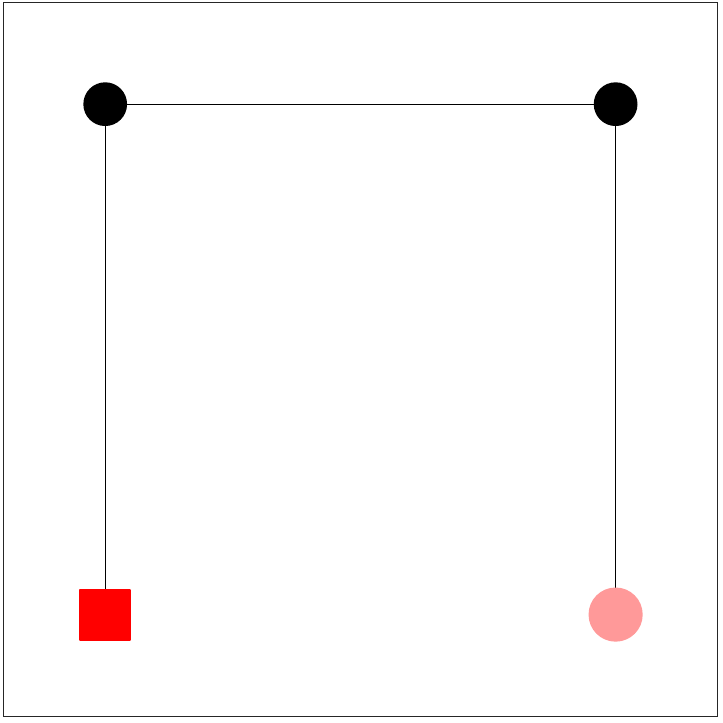}
  \caption*{$\boldsymbol\sigma_{\mathbf{\graphletID}}$}
\end{subfigure}
\hspace*{\fill}
 \renewcommand{\graphletID}{9}
\begin{subfigure}{\figGraphletWidth}
  \includegraphics[width=\linewidth]{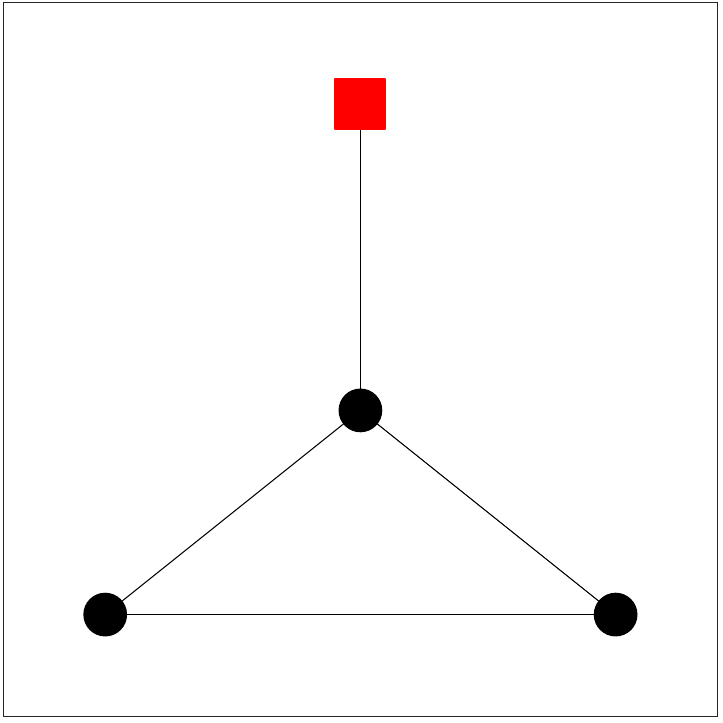}
  \caption*{$\boldsymbol\sigma_{\mathbf{\graphletID}}$}
\end{subfigure}
\hspace*{\fill}
 \renewcommand{\graphletID}{13}
\begin{subfigure}{\figGraphletWidth}
  \includegraphics[width=\linewidth]{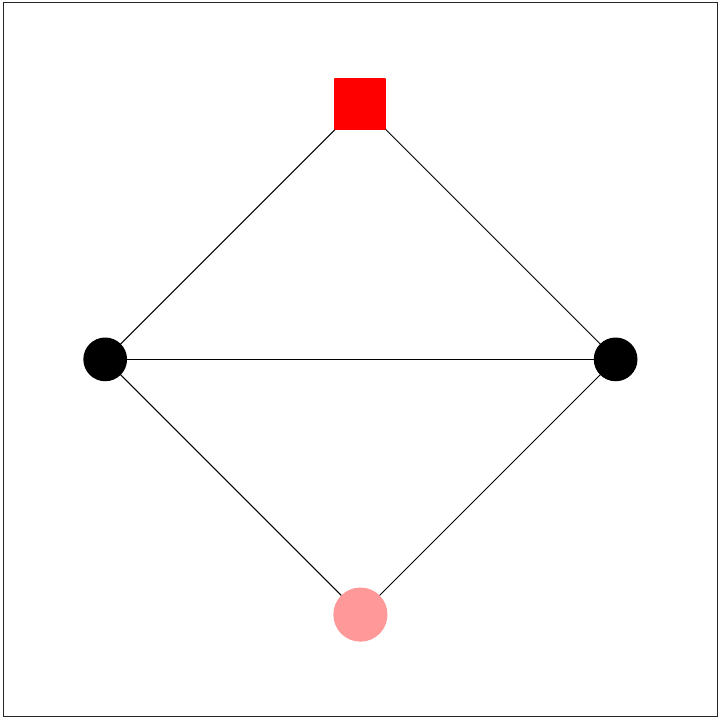}
  \caption*{$\boldsymbol\sigma_{\mathbf{\graphletID}}$}
\end{subfigure}
\hspace*{\fill}
 \\[0.5em]
\hspace*{\fill}
\renewcommand{\graphletID}{2}
\begin{subfigure}{\figGraphletWidth}
  \includegraphics[width=\linewidth]{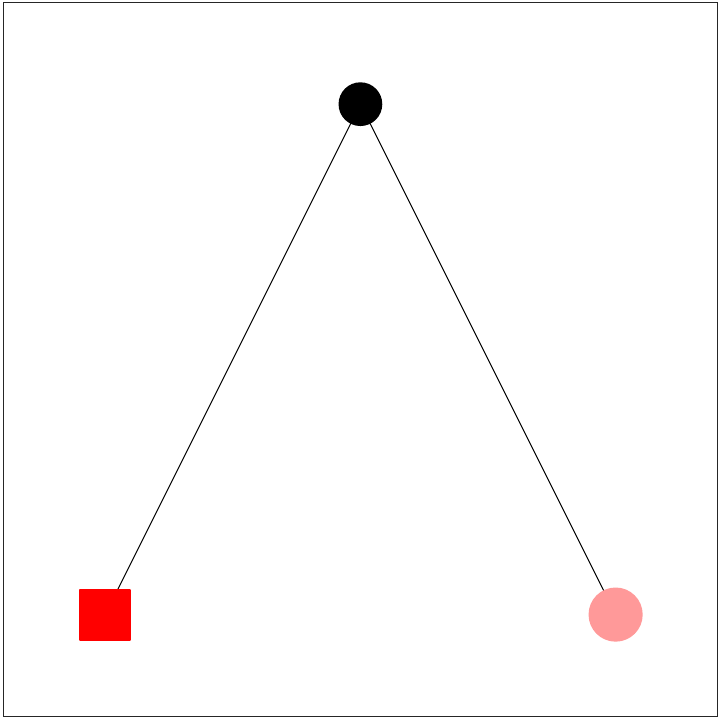}
  \caption*{$\boldsymbol\sigma_{\mathbf{\graphletID}}$}
\end{subfigure}
\hspace*{\fill}
 \renewcommand{\graphletID}{6}
\begin{subfigure}{\figGraphletWidth}
  \includegraphics[width=\linewidth]{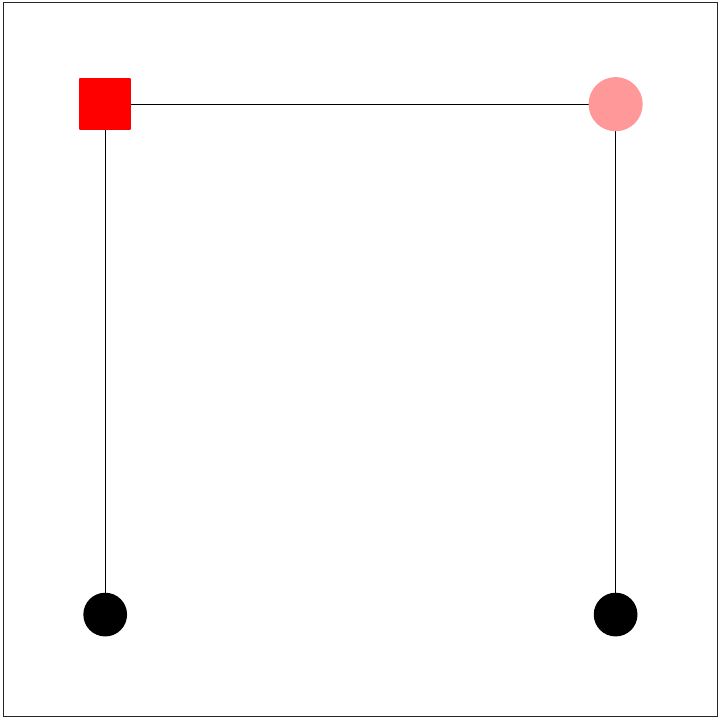}
  \caption*{$\boldsymbol\sigma_{\mathbf{\graphletID}}$}
\end{subfigure}
\hspace*{\fill}
 \renewcommand{\graphletID}{10}
\begin{subfigure}{\figGraphletWidth}
  \includegraphics[width=\linewidth]{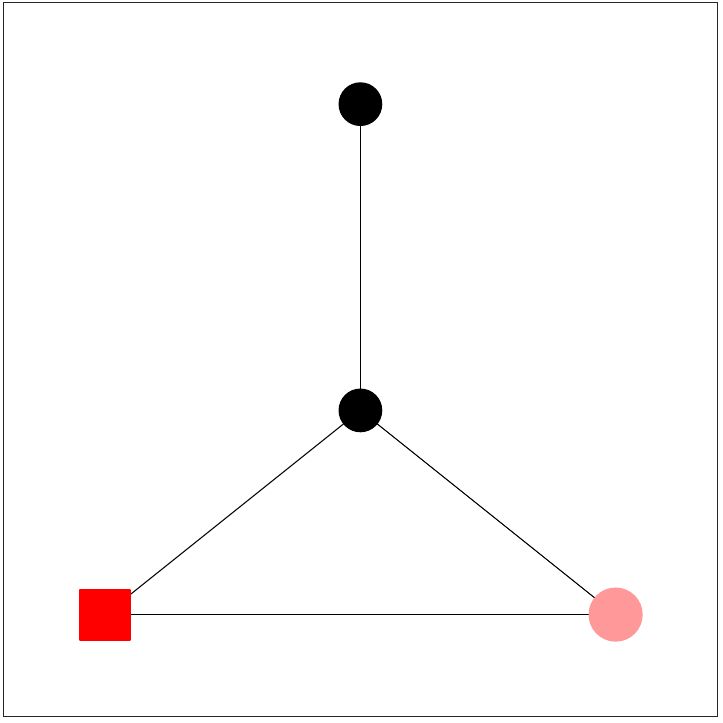}
  \caption*{$\boldsymbol\sigma_{\mathbf{\graphletID}}$}
\end{subfigure}
\hspace*{\fill}
 \renewcommand{\graphletID}{14}
\begin{subfigure}{\figGraphletWidth}
  \includegraphics[width=\linewidth]{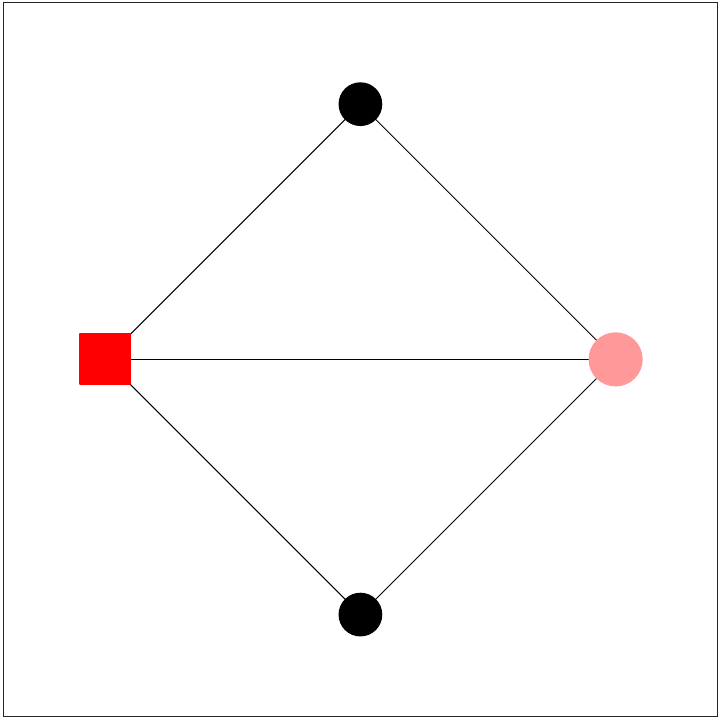}
  \caption*{$\boldsymbol\sigma_{\mathbf{\graphletID}}$}
\end{subfigure}
\hspace*{\fill}
 \\[0.5em]
\hspace*{\fill}
\renewcommand{\graphletID}{3}
\begin{subfigure}{\figGraphletWidth}
  \includegraphics[width=\linewidth]{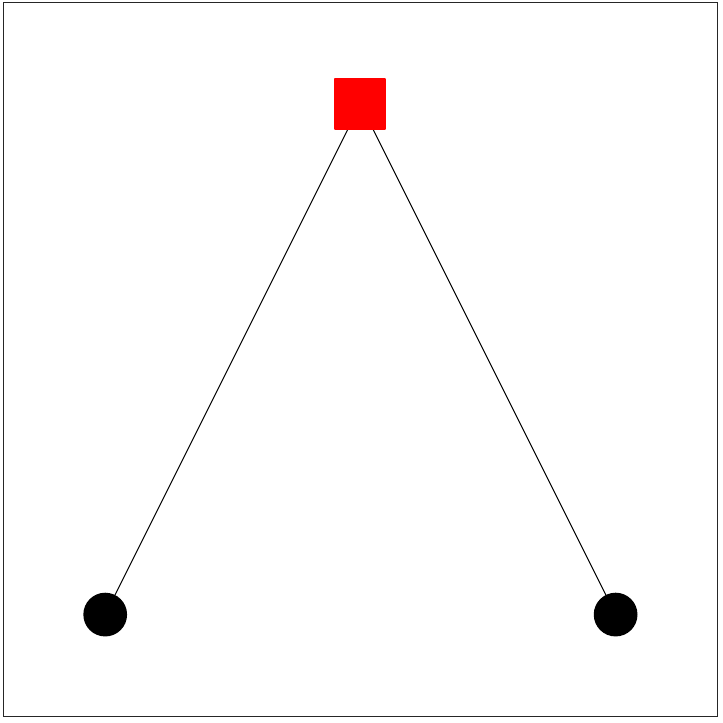}
  \caption*{$\boldsymbol\sigma_{\mathbf{\graphletID}}$}
\end{subfigure}
\hspace*{\fill}
 \renewcommand{\graphletID}{7}
\begin{subfigure}{\figGraphletWidth}
  \includegraphics[width=\linewidth]{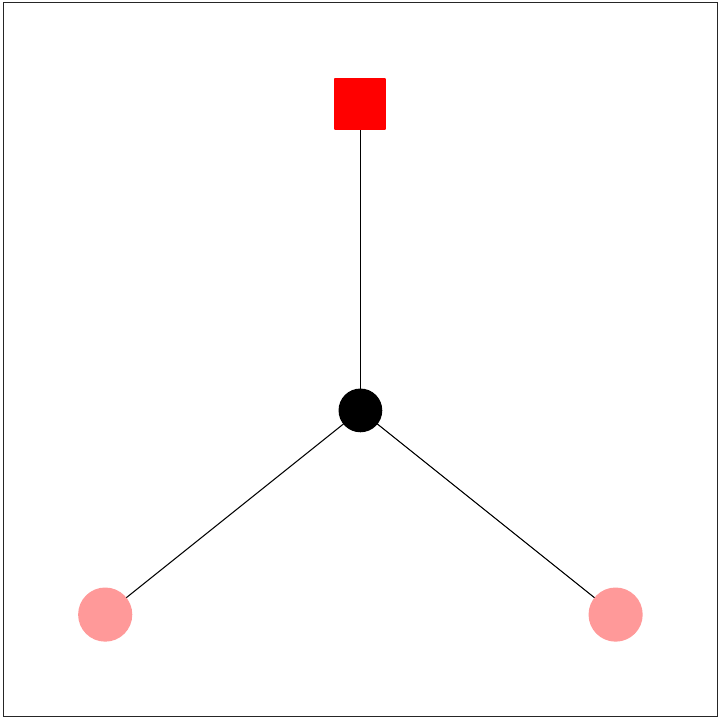}
  \caption*{$\boldsymbol\sigma_{\mathbf{\graphletID}}$}
\end{subfigure}
\hspace*{\fill}
 \renewcommand{\graphletID}{11}
\begin{subfigure}{\figGraphletWidth}
  \includegraphics[width=\linewidth]{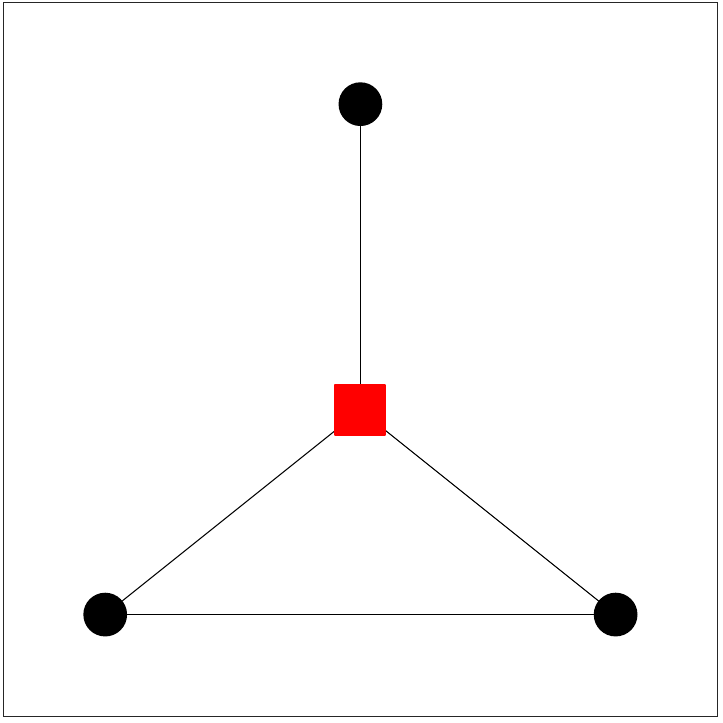}
  \caption*{$\boldsymbol\sigma_{\mathbf{\graphletID}}$}
\end{subfigure}
\hspace*{\fill}
 \renewcommand{\graphletID}{15}
\begin{subfigure}{\figGraphletWidth}
  \includegraphics[width=\linewidth]{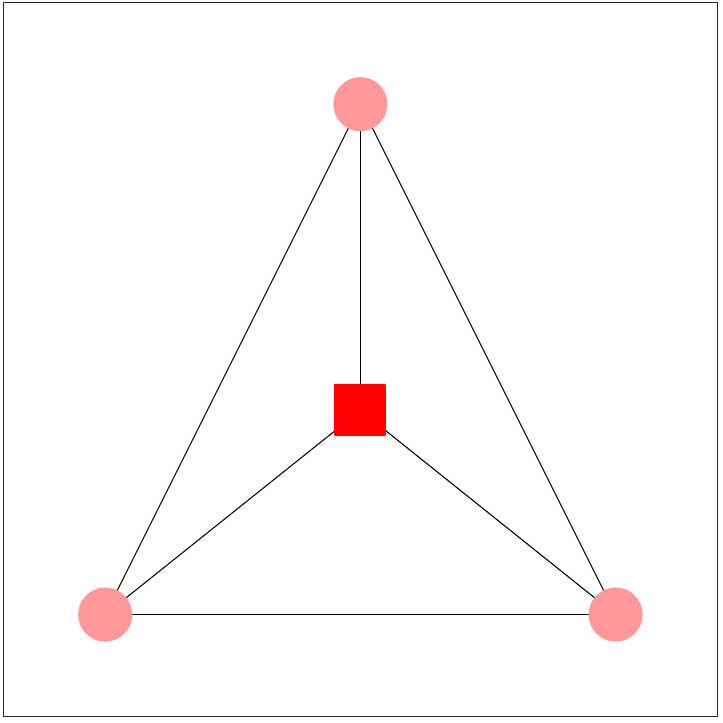}
  \caption*{$\boldsymbol\sigma_{\mathbf{\graphletID}}$}
\end{subfigure}
\hspace*{\fill}
 \\[0.5em]
   \caption{Dictionary $\Sigma_{16}$ of $16$ graphlets.  In each
    graphlet, the designated incidence node is specified by the red
    square marker, its automorphic position(s) specified by red
    circles.  The total ordering (labeling) of the graphlets is by the
    following nesting conditions. The graphlets are ordered first by
    non-decreasing number of vertices. Graphlets with the same vertex
    set belong to the same family.  Within each family, the ordering
    is by non-decreasing number of edges, and then by increasing
    degree at the incidence node (except the $4$-cycle).  The inclusion
    of $\sigma_{0}$ is necessary to certain vertex partition
    analysis~\cite{floros2020}.
  }  %
  \label{fig:graphlets}
\end{figure}

Network analysis using graphlets has advanced in recent years.  The
concepts of graphlets, graphlet frequency, and graphlet analysis are
originally introduced in 2004 by Pr\v{z}ulj, Corneil and Jurisica
\cite{przulj2004}. They have been substantially extended in a number of ways
\cite{yaveroglu2015,sarajlic2016,shervashidze2009a,przulj2019,windels2018}.
Graphlets are mostly used for statistical characterization and
modeling of entire networks. In the work by Palla
et. al. \cite{palla2005}, which is followed by many, a network of
motifs (a special case of graphlets) is induced for overlapping
community detection on the original network.
Recently we established a new way of using graphlets for graph
analysis. We use graphlets as coding elements to encode topological
and statistical information of a graph at multiple granularity levels,
from micro-scale structures at vertex neighborhoods,
up to macro-scale structures such as cluster
configurations~\cite{floros2020}. We also use the topology
encoded information to uncover temporal patterns of variation and
persistence across networks in a time-shifted sequence, not
necessarily over the same vertex set~\cite{floros2020a}.

We anticipate a growing interest in, and applications of,
graphlet-based network/graph analysis, for the following reasons.
Graphlets are fundamental topology elements of all networks or
graphs. See a particular graphlet dictionary shown in
\Cref{fig:graphlets}.
Conceptually, graphlets for network/graph analysis are similar to
wavelets for spectro-temporal analysis in signal
processing~\cite{rioul1991}, shapelets for time series
classification~\cite{ye2009}, super-pixels for image
analysis~\cite{ren2003}, and n-grams for natural language
processing~\cite{shannon1948,shannon1951}.
Like motifs, graphlets are small graphs.  By conventional definition,
motifs are small subgraph patterns that appear {presumptively and
significantly } more frequently in a network under study. Motif analysis
relies on prior knowledge or assumption~\cite{milo2002}.
Graphlets are ubiquitous; graphlet analysis reveals the most frequent
connection patterns or motifs, or lack of dominance by any, in a
network.

There is another aspect of the universality in using graphlets.
Graphlets are defined in the graph-topology space.  They are not to be
confused with the wavelets applied to the spectral elements of a
particular graph Laplacian as in certain algebraic graph analysis. The
latter is limited to the family of graphs defined on the same vertex
set and share the same eigenvectors. Otherwise, the Laplacians of two
graphs on the same vertex set are not commutable.  The computation of
Laplacian spectral values and vectors is also limited, by complexity
and resources, to low-dimensional invariant subspaces. Graphlets hold
a promise to overcome the limitation.

The time and space complexities of graph encoding with graphlets,
i.e., the graphlet transform, have not been formally described and
addressed. The transform with encoding dictionary $\Sigma$, to be
described in \Cref{sec:graphlet-transform-problem-description}, maps
graph $G=(V,E)$ to a $|V|\times |\Sigma|$ array of graphlet
frequencies at all vertices.  In fact, the mapping is related to the
classical problem of finding, classifying and counting small subgraphs
over vertex neighborhoods \cite{przulj2006,duke1995}. A familiar case
is to find and count all triangles over the entire graph. The triangle
is graphlet $\sigma_{4}$ ($C_3, K_3$) in \Cref{fig:graphlets}.  With
the graphlet transform, the number/frequency of distinct triangles
incident on each and every vertex is computed.
It is found, from an analysis of scientific collaboration
networks~\cite{floros2020a}, that the bi-fork graphlet $ \sigma_{3} $
($K_{1,2}$) encodes the betweenness among triangle clusters.
Another familiar case is to find and count induced claw subgraphs, the
claw is graphlet $\sigma_{8}$ ($K_{1,3} $).  The naive method checks
every connected quad-node subgraph for claw recognition. Its time
complexity is $O(n^4)$, $n=|V|$.  The naive method can be accelerated
by applying fast matrix multiplication algorithms on asymptotically
sufficiently large graphs, at the expense of greater algorithmic
complication, and if feasible to implement, with increased memory
consumption, loss of data locality and increased latency in memory
access on any modern computer with hierarchical memory. There is
another type of counting methods that search the patterned subgraphs
from neighborhood to neighborhood, with detailed book
keeping~\cite{chiba1985,kloks2000}.
Due to the high computational complexity, certain network analysis
with graphlets resorts to nondeterministic approximation with sparse
sampling under a structure-persistent
assumption~\cite{przulj2006}.  Otherwise, it is known that certain
network properties are not preserved with sampling~\cite{stumpf2005}.

This work makes a few key contributions.  We formally introduce the
graphlet transform problem, and address the issues with encoding
capacity and complexity. We present sparse and fast formulas for the
graphlet transform of any large, sparse graph, with any sub-dictionary
of $\Sigma_{16}$ as the coding basis.  The transform is deterministic,
exact and directly applicable to any range of graph size.  Our
solution method establishes remarkable records at once in multiple
aspects -- time complexity, memory space complexity, program
complexity and high-performance implementation.
Particularly, the time complexity of the fast graphlet transform with
any dictionary $\Sigma \!\subseteq\! \Sigma_{16}$ is linear in $(|V| +
|E|) |\Sigma|$ on degree-bounded graphs or planar graphs.
Contrary to existing methods, the transform formulas can be
straightforwardly translated to high performance computation\cite{floros2020c},
via the use of readily available software libraries such as {\tt
GraphBLAS}~\cite{davis2018}.
We also address criteria of selecting graphlet elements for certain
counting-based decision or detection problems on graphs.
This work serves twofold objectives: to enable large network/graph
analysis with graphlets and to enrich and advance sparse graph theory,
computation and their applications.

The basic assumptions and notations throughout the rest of the paper
are as follows.  Graph $G=(V,E)$ has $n=|V|$ nodes/vertices and
$m=|E|$ edges/links.  It is sparse, such as $m = O(n\log^{k}n)$ with a
small value of $k$.  The nodes are indexed from $1$ to $n$, a
particular ordering is specified when necessary.
Graph $G$ is simple, undirected and specified by its (symmetric)
adjacency matrix $A$ of 0-1 values.  The $j$-th column of $A$, denoted
by $a_j$, marks the neighbors of node $j$. Denote by $e_j$ the $j$-th
column of the identity matrix. The sum of all $e_j$ is the constant-1
vector, denoted by $e$. The maximal degree is $d_{\max}$.
The Hadamard (elementwise) multiplication is denoted by $\odot$.  The
number of nonzero elements in matrix $B$ is $\mathrm{nnz}(B)$.  The
total number of arithmetic operations for constructing $B$ is
$\mathrm{cost}(B)$.
For any two non-negative matrices $A$ and $B$, $A\!-\! B$ is the
shorthand expression for the sparse difference, i.e., the elementwise
rectified difference $ \max\{ A-B, 0\}$.

\section{Graphlet transform: problem description}
\label{sec:graphlet-transform-problem-description}

We now describe generic graphlets and graphlet dictionaries by their
forms and attributes, with modification in description over the
original, for clarity.
A {\em graphlet} is a connected graph with a small vertex set and a
unique orbit ( a subset of vertices symmetric under permutations).  We
show
in \Cref{fig:graphlets}  %
a dictionary of $16$ graphlets,
$\Sigma \! =\! \Sigma_{16} \!=\! \{\, \sigma_k \}_{k=0:15}$.
The graphlets in the dictionary have the following patterns:
singleton/vertex, edge ($K_{2} $), 2-path ($P_2$), binary fork
($K_{1,2}$), triangle ($C_3$, $K_3$), 3-path ($ P_3 $), binary fork
($K_{1,2}$), claw ($ K_{1,3} $), paw ($(3,1)$- tadpole), 4-cycle
($ C_4 $), diamond ($K_{1,1,2}$), and tetrahedron ($ K_{4} $).
Each graphlet has a designated incidence node, shown with a red
square, unique up to an isomorphic permutation (shown in red circles).
In short, $\Sigma_{16}$ contains all connected graphs up to $4$ nodes
with distinctive vertex orbits. 
Graphlets on the same vertex set form a family with an internal
partial ordering.  For example, in the tri-node family, the partial
ordering $\sigma_2, \sigma_3 \prec \sigma_4$ denotes the relationship
that $\sigma_{2}$ and $\sigma_{3}$ are subgraphs of $\sigma_{4}$.
Our rules for the ordering/labeling are described in the caption of
\Cref{fig:graphlets}, for convenience in visual verification.

We use a vertex-graphlet incidence structure to describe
the process of encoding $G$ over the entire vertex set $V$ with coding
elements in $\Sigma$.
Let $G=(V,E)$ be a graph. Let $\Sigma$ be a graphlet dictionary, the
code book. 
Denote by $ B = ( V, \Sigma ; E_{v\sigma}) $ the bipartite between the
graph vertices and the graphlets,
$E_{v\sigma} \subset V\times \Sigma$.
There is a link $(v,\sigma) $ between a vertex $v \in V$ and a
graphlet $\sigma \in \Sigma$ if $v$ is an incident node on a subgraph
of $\sigma$-pattern.
The incident node on a graphlet is uniquely specified, up to an
isomorphic mapping.
For example, graphlet $\sigma_6$ (clique $K_4$) in
\Cref{fig:graphlets} is an automorphism. 
There may be multiple links between $v$ and $\sigma_k$.  We denote
them by a single link $(v, \sigma_k)$ with a positive integer weight
$d_k(v)$ for the multiplicity, which is the frequency with graphlet
$\sigma_k$.
However, the multiplicities from vertex $v$ to multiple graphlets in
the same family are not independently determined.
For example, the multiplicities on links from vertices to $\sigma_2$
do not include those within $\sigma_4$.  The weight on
$ (v, \sigma_1) $ is counted independently as $\sigma_1$ has no
other family member.
For any vertex, $d_{0}(v) = 1$, $d_1(v)$ is the ordinary degree of $v$
on graph $G$. With $k>1$, $d_k(v)$ is a pseudo degree, depending on
the internal structure of the family $\sigma_k$ is in.
This vertex-graphlet incidence structure is a generalization of the
ordinary vertex-edge incidence structure.

The {\em graphlet transform} of graph $G$ refers to the mapping
$f$ of $G$ to the field of graphlet frequency vectors over
$V$,
\begin{equation}
\label{eq:graphlet-transform}
f( v ) = [\, d_{0}(v), d_{1}(v), \cdots, d_{|\Sigma|-1}(v) \, ]^{\rm T},
\quad v \in V. 
\end{equation}
The vector field encodes the topological and statistical information of the
graph. The transform is orbit-invariant, i.e., for $u$ and
$v$ on the same orbit, $f(u)=f(v)$.  It is graph invariant, i.e., for
isomorphic graphs $G$ and $G'$, $f(G) = f(G')$.
We introduce how we can make this transform fast.

Consider first the coding capacity.  The dictionary
$\Sigma_{2} = \{ \sigma_{0}, \sigma_{1} \}$ is the minimal. It limits
the network analysis to the ordinary degree distributions, types,
correlations and models~\cite{barabasi2016,newman2011,posfai2013}.
The dictionary $\Sigma_{5}$, a sub-dictionary of $\Sigma_{16}$,
already offers much greater coding capacity. 
We answer an additional, interesting question -- how the
computation complexity changes with the coding element selection.

\section{Fast graphlet transform with $ \Sigma_{5} $}
\label{sec:transform-with-Sigma5}

\subsection{Preliminary lemmas}
\label{sec:paths-cycles}

We start with graphs of paths and graphs of cycles, using matrix
expressions and operations.
Denote by $G(P_{\ell})$ the graph of length-$\ell$ paths over $G$ with
weighted adjacency matrix $P_{\ell}$, $\ell > 0$.  Element
$P_{\ell}( i,j) $ is the number of length-$\ell$, simple (i.e.,
loop-less) paths between node $i$ and node $j$.
Let $p_{\ell} = P_{\ell} \, e$. It represents the scalar function on
$V$ such that $ p_{\ell}(i) $ is the total number of length-$\ell$
paths with node $i$ at one of the ends.  In particular, $P_{1} = A$,
$p_{1} = d_{1} $. We have
\begin{equation}
  \label{eq:P2}
  P_{2} = A^{2} - \mbox{diag}( d_{1} ), 
\end{equation}
where `diag' denotes the construction of a diagonal matrix. 
We describe the following important fact.
\begin{lemma}[Matrix of $2$-paths.]
  \label{lemm:P2-formation}
  Matrix $P_2$ is the accumulation of 2-column contribution from each
  and every edge,
  \begin{equation}
    \label{eq:A2-formation}
    P_{2}= \sum_{(i,j)\in E} (a_i-e_{j}) e_j^{\rm T} + (a_j-e_i)e_i^{\rm T}. 
    \end{equation} 
    Consequently,
$ \mathrm{nnz}( P_{2}) \leq  \mbox{\rm cost}(P_2) < 2 \cdot d_{\max} \cdot m. $ 
Similarly, $A^2 = \sum_{(i,j) } a_{i}e_{j}^{\rm T} + a_{j}e_{i}^{\rm T} $.
\end{lemma}

Denote by $G( C_{\ell})$ the graph of length-$\ell$ cycles over $G$
with weighted adjacency matrix $C_{\ell}$, $\ell > 1 $.  Element
$C_{\ell}( i,j )$ is the number of length-$\ell$ simple cycles that
pass through both $i$ and $j$.
We denote by $c_{\ell}$ the vertex function on $V$ such that $c_{\ell}(i)$
is the total number of length-$\ell$ simple cycles passing through
node $i$.  A simple cycle is a simple path that starts from and ends
at the same node.
  \begin{lemma}[Sparse graph of cycles]  
      \label{lemm:cycle-sparse}
      For $\ell > 1$, matrix $C_{\ell}$ is as sparse as $A$, 
  \begin{equation}
    \label{eq:sparse-cycle-graph}
    C_{\ell}   =   A \odot P_{\ell-1} ,
  \end{equation}
  Additionally, $ c_{\ell} = C_{\ell}\, e/(\ell\!-\! 1)$.
  In parituclar, $c_{2} = d_{1} $, $C_{3} = A\odot A^{2}$ and
  $c_{3} = C_{3}\, e / 2$.
\end{lemma} 
A consequence of \Cref{eq:sparse-cycle-graph} is an alternative
formulation of $p_2$ without forming matrix $P_2$:
\begin{equation} 
\label{eq:p2} 
                 p_2 = A\, p_1 - c_{2},  
\end{equation}
The next lemma also has a key role in complexity analysis in the rest
of the paper. A proof is in Appendix~\ref{apdix:K3-complexity}.
\begin{lemma}[Triangle count and counting cost] 
  \label{lemm:triangle-total}
  The total number of triangles is $e^{\rm T} C_3 e/6$. 
  Denote by $\mathrm{cost}(c_3)$ the cost for computing
  $c_3=C_3e/2$. Then,
  \begin{equation}
    \label{eq:K3-Ubound}
     e^{\rm T} C_3 e  \leq \mathrm{cost}(c_3)
    \displaystyle 
      \leq \min\{ d_{\max}, 2\alpha(G) \}\, m,
\end{equation}
where $\alpha(G)$ is the arboricity of graph $G$~\cite{nash-williams1961}. 
\end{lemma}

\subsection{Tri-node graphlet frequencies}

There is a partial ordering among the three members of the tri-node family,
\begin{equation}
  \label{eq:tri-node-graphlet-ordering}
  \sigma_{2}, \sigma_3 \prec \sigma_{4},
\end{equation}
by the relationship that $\sigma_{4}$, the triangle, has $\sigma_2$
and $\sigma_{3}$ as subgraphs.  The frequency with $\sigma_2$ at node
$i$ in graph $G$ does not include those $\sigma_2$ subgraphs in any
triangle. Similarly with the $\sigma_{3}$ frequency at any node.

We have by now the vectors $d_{0}=e$, $d_{1}=A\, e$ and $d_{4} = c_3$. 
It is straightforward to verify that $d_{2} = p_2 - c_{3}$.  We have
the following expression for the bi-fork graphlet frequency vector.
   \begin{equation}
     \label{eq:bi-fork-frequencies} 
    d_{3}  = p_1 \odot (p_1 - 1) / 2  - 2\, c_3.
\end{equation}
\begin{theorem}[Fast graphlet transform with $\Sigma_5$] 
  The graphlet transform of $G=(V,E)$ with $\Sigma_{5}$ takes
  no more than 
  $ 3\,  \min\{  d_{\max}, 2\alpha(G) \} \, m $  arithmetic
  operations and $ 6\,  ( m + n )$ memory space.
\end{theorem}

\section{Fast graphlet transform with $ \Sigma_{16} $}
\label{sec:transform-with-Sigma15}

We turn our attention to the family of quad-node graphlets.  The
family has $11$ members ($\sigma_{5}$ to $\sigma_{15}$) with the
following partial ordering in terms of subgraph relationship,
\begin{equation}
  \begin{aligned}
    \label{eq:quad-node-graphlet-ordering}
    & \sigma_{5}, \, \sigma_{6}
    \prec
    \sigma_{9}, \, \sigma_{10}, \, \sigma_{11}, \, \sigma_{12} ; 
    \\
    &  \sigma_{7}, \, \sigma_{8}    
    \prec
    \sigma_{9}, \, \sigma_{10}, \, \sigma_{11} ;
    \\
    &
    \sigma_{9},  \sigma_{10}, \, \sigma_{11}, \, \sigma_{12}
    \prec    \sigma_{13}, \sigma_{14} ; 
    \\
    & \sigma_{13}, \, \sigma_{14}
    \prec \sigma_{15}. 
  \end{aligned}
\end{equation}

With each graphlet $\sigma_i$ we derive first the formula for its {\em
  raw} or independent frequency at vertex $v$, denoted by
$\hat{d}_{i}(v)$, as the number of $\sigma$-pattern subgraphs incident
with $v$. The subgraphs include the induced ones. The raw frequency
vector is 
$ \hat{f} (v) = [\, \hat{d}_0(v), \hat{d}_{1}(v), \cdots,
\hat{d}_{|\Sigma|-1} (v)\,
]^{\rm T}$.
We will then convert the raw frequencies to the nested, or {\em net},  
frequencies of \Cref{eq:graphlet-transform}.  
The net frequencies depend on the inter-relationships between the
graphlets in a dictionary, as shown by the partial ordering in
\cref{eq:quad-node-graphlet-ordering} for $\Sigma_{16}$.
We always have $ \hat{f}(v) \geq f(v)$.
We shall clarify the connection between net frequencies and induced
subgraphs. When, and only when, the family of $k$-node graphlets is
complete with distinctive connectivity patterns and orbits, and
non-redundant, the net frequency of graphlet $\sigma$ at vertex $v$ is
the number of $\sigma$-pattern {\em induced} subgraphs incident with
$v$.
For instance, a 3-star (claw) subgraph in a paw is not the induced
graph by the same vertex set. Under the complete and non-redundant
family condition, the frequency conversion has the additional
functionality to identify precisely the patterns of induced subgraphs.
We will describe in \Cref{sec:frequency-conversion} a unified scheme
for converting raw frequencies to net ones.
The dependencies within graphlet families can be relaxed for graph
encoding purposes other than pattern recognition.
We derive fast formulas for quad-node graphlets in $3$
subgroups.

\subsection{Frequencies of paths \& cycles} 
\label{sec:path3-cycle4} 

We relate the frequencies with 3-path graphlet $\sigma_{5}$ and
gate graphlet $\sigma_{6}$ to that with $p_1$ and $p_2$.  The
following are straightforward,
\begin{equation}
    \label{eq:d5-d6}
      \hat{d}_{5} = p_3, 
      \quad 
      \hat{d}_{6} = p_{2} \odot ( p_{1} - 1 ) - 2\, c_3. 
  \end{equation} 
\begin{lemma} [Fast calculation of 3-path frequencies] 
  \begin{equation}
    \label{eq:p3}
    p_{3} = A \, p_{2} - p_{1} \odot (p_{1}-1) - 2\, c_{3}  . 
  \end{equation}    
  \end{lemma}
\begin{proof}
  We get $p_3 = P_3e$ by the expression of  
  \begin{equation}
  \label{eq:P3}
  P_{3} = A \, P_2 - \mathrm{diag}(p_1-1) \,
    P_1 - 2\, \mbox{diag}(c_3),
\end{equation}
where we extend $P_2$ by one step walk, remove $1$-step backtrack, and
remove triangles on the diagonal.
\end{proof}
By the lemma, vector $p_3$ is obtained without formation of $P_3$,
which invokes the cubic power of $A$.
Next, we obtain vector $c_4$ without constructing $C_4 = A\odot P_3$
of \Cref{lemm:cycle-sparse}. %
\begin{lemma} [Fast calculation of 4-cycle frequencies] 
\label{lemm:c4}
Denote by $G(C_{4,2})$ the graph with adjacency matrix $ C_{4,2} $
such that element $C_{4,2}(i,j)$ is the number of distinct $4$-cycles
passing through two nodes $i$ and $j$ at diametrical positions. Then,
\begin{equation}
  \label{eq:c4}
    C_{4,2} = P_{2} \odot  ( P_{2} - 1  ) , 
    \quad 
    c_{4}  = C_{4,2} \, e / 2.
\end{equation}
Consequently, $\mathrm{nnz} ( C_{4,2} ) \leq \mathrm{nnz} (P_{2} )$.
\end{lemma}
By the diametrical symmetry, $\sum_{j} C_{4,2}(i,j) $ is twice the
total number of $4$-cycles passing through $i$.

The essence of the fast frequency calculation lies in constructing
sparse auxiliary matrices and vectors which use Hadamard products for
both logical conditions and arithmetic operations, without
confining/limiting to logical operations (such as in circuit
expressions), to arithmetic operations (such as in methods using fast
matrix-matrix products), or to local spanning operations. In the same
vein, we present formulas for fast calculation of the remaining $8$ graphlet
frequencies in brief statements and proof sketches.

\subsection{Frequencies of claws \& paws}
\label{sec:claws-paws}

This section contains fast formulas for raw frequencies with two
claw graphlets and three paw graphlets.
Graphlet $\sigma_{7}$ is the claw with the incidence node at a leaf
node. At node $i$, we sum up the bi-fork counts over its $p_{1}$
neighbors, excluding the one connecting to $i$, i.e., $(p_1(i)\!-\! 1)$
choose $2$. Thus,
\begin{equation}
  \hat{d}_{7} = A \, \big( (p_1 - 1) \odot (p_1 -
  2) \big) / 2.
\end{equation}

Graphlet $\sigma_{8}$ is the claw ($K_{1,3}$) with the incidence node
at the root/center. We have 
\begin{equation}
  \hat{d}_{8} =  p_{1}  \odot (p_1 - 1) \odot ( p_1 - 2 )/ 3!, 
  \label{eqn:sigma-claw-center}
\end{equation}
by the fact that the number of 3-stars centered at a node $i$ is
$p_{1}(i)$ choose $3$.
  
Graphlet $\sigma_{9}$ is the paw with the incidence node at
the handle end. We have  
\begin{equation}
  \hat{d}_{9} =  A \, c_{3} - 2\,  c_{3}.  %
\end{equation}
The triangles passing $i$ are removed from the total number of triangles incident
at the neighbor nodes of $i$.

Graphlet $\sigma_{10}$ is the paw with the incidence node at a
base node. We have 
\begin{equation}
  \hat{d}_{10} = C_3 \, (p_1 - 2).
\end{equation}
Each triangle at node $i$ is multiplied by the number of other
adjacent nodes that are not on the same triangle.  By \Cref{lemm:cycle-sparse},
$C_3$ is as sparse as $A$.

Graphlet $\sigma_{11}$ is the paw with the incidence node at
the center (degree 3). We have 
\begin{equation}
  \label{eq:4}
  \hat{d}_{11} = (p_1 -2) \odot c_3.
\end{equation}
At the incident node $i$, the number of triangles is multiplied by all
other edges leaving node $i$.

For this group of graphlets, the calculation of the raw frequencies
uses either vector operations or matrix-vector products with either
$A$ or a matrix as sparse as $A$.

\subsection{Frequencies of diamonds \& tetrahedra}  

Graphlet $\sigma_{13}$ is the diamond with the incidence node at an
off-cord node $i$. We have 
\begin{equation}
    D_{4,c} \triangleq A\odot ( A (C_{3}-A) ),  
    \quad 
    \hat{d}_{13} = D_{4,c}\, e/2.
  \label{eqn:sigma-diamond-off-cord}
\end{equation}
The element $D_{4,c}(i,j)$ is the number of diamonds with
off-cord node $i$ and on-cord node $j$.  
\begin{proof}
  With $i$ as an off-cord node, its on-cord neighbors must form a
  triangle with $i$ and a triangle with another node. Thus, the
  account on an off-cord node $i$ is
  $a_i^{\rm T}( C_{3} - 1) a_{i} = a_i^{\rm T}( C_{3} - A) a_{i}$, or
  equally, $ ( a_i^{\rm T} \odot (a_{i}^{\rm T} (C_3-A) ) e$.
\end{proof}  

Graphlet $\sigma_{14}$ is the diamond with the incidence node at a
cord node. We have
\begin{equation}
  \label{eqn:sigma-diamond-on-cord}
  D_{4,3} \triangleq A\odot C_{4,2},
   \quad   
  \hat{d}_{14} = D_{4,3} \, e / 2.
\end{equation}
The Hadamard product is sparse, $C_{4,2}$ is defined in
\Cref{lemm:c4}.
\begin{proof} 
  Node $i$ on a 4-cycle must be connected with its diametrical node.
\end{proof}

Graphlet $\sigma_{15}$ is clique $K_4$.  Define matrix
$ T $ as follows,
\begin{equation}
  \label{eq:K4-matrix} 
  T \triangleq A \odot \big[ q_{ij}^{\rm T}  A\, q_{ij}, \, (i,j) \in E \big],
  \quad q_{ij} = a_i \odot a_{j}, 
\end{equation}
where $a_j = A\, e_j$. We have
\begin{equation}
  \label{eq:K4-frequency}
  \hat{d}_{15}  = T \, e /6 .
\end{equation}
\begin{proof} 
  Vector $q_{ij}$ indicates the common neighbors between nodes $i$ and
  $j$.
  When $ q_{ij}(k) q_{ij} (\ell) A(k, \ell) \neq 0$, the subgraph at
  $\{ i,j,k,\ell \} $ is a tetrahedron. The total number of distinct
  tetrahedra incident with edge $ (i,j)$ is
  $ T( i, j ) = \sum_{k>\ell} q_{ij}(k) q_{ij}(\ell) A (k,\ell)/3 =
  q_{ij}^{\rm T} A q_{ij}/6$.
\end{proof}
Matrix $T$ is sparse. For $(i,j)\in E$, computing $T(i,j)$ takes
no more than $3\, \mathrm{nnz}(q_{ij})^2$ arithmetic operations.

\section{A unified scheme for frequency conversion}
\label{sec:frequency-conversion}

We summarize in \Cref{tab:summary-formulas} the formulas in
matrix-vector form for fast calculation of the raw frequencies. The
auxiliary vectors are $p_j$ and $c_j$, $1\leq j \leq 4$, each of which
is elaborated in
\cref{sec:transform-with-Sigma5,sec:transform-with-Sigma15}. The
auxiliary matrices $C_3$, $A\odot C_{4,2}$, $D_{4}$ and $T$ are as
sparse as $A$.

We provide in \Cref{tab:frequency-conversion-matrix} the (triangular)
matrix $ U_{16} $ of nonnegative coefficients for mapping net
frequencies $d(v)$ to raw frequencies $ \hat{d}(v) $.  The conversion
coefficients are determined by subgraph-isomorphisms among graphlets
and automorphisms in each graphlet. The frequency conversion for any
sub-dictionary of $\Sigma_{16}$ is by the corresponding sub-matrix of
$U_{16}$.  We actually use the inverse mapping to filter out
non-induced subgraphs. The conversion matrix, invariant across the
vertices, is applied to each and every vertex.  The conversion
complexity is proportional to the product of $|V|$ and the number of
nonzero elements in the conversion matrix.  The number of nonzero
elements in $U_{16}$ is less than $3|\Sigma_{16}|$. The inverse
$U_{16}^{-1}$ has exactly the same sparsity pattern as $U_{16}$. The
identical sparsity property also holds between each sub-dictionary
conversion matrix and its inverse.

We illustrate in \cref{fig:briki} the graphlet transform of a small
graph $ G=(V,E) $ with $6$ vertices and $9$ edges.  With each graphlet
$\sigma_{i}$, the raw frequencies $\hat{d}_i$ across all vertices are
calculated by the fast formulas in \Cref{tab:summary-formulas} and
tabulated in the top table/counts, computed by our fast transforms.
The $i$-th row in the table is the raw frequency vector
$\hat{f}(v_i)$.  The raw frequency vectors are converted to
the net frequency vectors $ \{ f(v) , v \in V \}$ of
\Cref{eq:graphlet-transform} by matrix-vector multiplications with the
same triangular matrix $U^{-1}_{16}$.
As the fast graphlet transform is exact, we made accuracy comparison
between the results by our sparse and fast formulas and that by the
dense counterparts. The results are in full agreement.
The transform has additional values in systematic quantification and
recognition of topological properties of the graph, as briefly noted
in the caption of \cref{fig:briki}.

\begin{figure*}[tbh]
  \centering
  \begin{subfigure}{0.38\linewidth}
    \centering
    \begin{tikzpicture}
      [scale=.4,auto=left,every node/.style={circle, minimum size=#1, fill=blue!20}]
      \node (n6) at (1,6) {6};
      \node (n4) at (4,4)  {4};
      \node (n5) at (8,5)  {5};
      \node (n1) at (11,4) {1};
      \node (n2) at (9,2)  {2};
      \node (n3) at (5,1)  {3};
      
      \foreach \from/\to in {n6/n4,n4/n5,n5/n1,n1/n2,n2/n5,n2/n3,n3/n4,n3/n5,n2/n4}
      \draw (\from) -- (\to);
    \end{tikzpicture}
  \end{subfigure}
  \hskip 1em
  \begin{subfigure}{0.58\linewidth}
    \centering
    \resizebox{\linewidth}{!}{%
      \begin{tabular}{lcccccccccccccccc}
        \toprule
        $v$
        & $\hat{d}_0$    & $\hat{d}_1$ & $\hat{d}_2$          & $\hat{d}_3$
        & $\hat{d}_4$    & $\hat{d}_5$ & $\hat{d}_6$          & $\hat{d}_7$
        & $\hat{d}_8$    & $\hat{d}_9$ & $\hat{d}_{10}$      & $\hat{d}_{11}$
        & $\hat{d}_{12}$ & $\hat{d}_{13}$ & $\hat{d}_{14}$  & $\hat{d}_{15}$ \\
        \midrule
        1 & 1 & 2 & 6 & 1 & 1 & 14 & 4 & 6 & 0 & 6 & 4 & 0 & 2 & 2 & 0 & 0 \\
        2 & 1 & 4 & 9 & 6 & 4 & 12 & 19 & 7 & 4 & 3 & 12 & 8 & 5 & 3 & 5 & 1 \\
        3 & 1 & 3 & 9 & 3 & 3 & 14 & 12 & 9 & 1 & 5 & 12 & 3 & 4 & 4 & 3 & 1 \\
        4 & 1 & 4 & 8 & 6 & 3 & 12 & 18 & 7 & 4 & 5 & 10 & 6 & 4 & 4 & 3 & 1 \\
        5 & 1 & 4 & 9 & 6 & 4 & 12 & 19 & 7 & 4 & 3 & 12 & 8 & 5 & 3 & 5 & 1 \\
        6 & 1 & 1 & 3 & 0 & 0 & 8 & 0 & 3 & 0 & 3 & 0 & 0 & 0 & 0 & 0 & 0 \\
        \bottomrule
      \end{tabular}%
    }
    \\[0.5em]
    \resizebox{\linewidth}{!}{%
      \begin{tabular}{lcccccccccccccccc}
        \toprule
        $v$
        & $d_{0}$ & $d_{1}$ & $d_{2}$ & $d_{3}$
        & $d_{4}$ & $d_{5}$ & $d_{6}$ & $d_{7}$
        & $d_{8}$ & $d_{9}$ & $d_{10}$ & $d_{11}$
        & $d_{12}$ & $d_{13}$ & $d_{14}$ & $d_{15}$ \\
        \midrule
        1 & 1 & 2 & 4 & 0 & 1 & 2 & 0 & 0 & 0 & 2 & 0 & 0 & 0 & 2 & 0 & 0 \\
        2 & 1 & 4 & 1 & 2 & 4 & 0 & 1 & 0 & 0 & 0 & 2 & 1 & 0 & 0 & 2 & 1 \\
        3 & 1 & 3 & 3 & 0 & 3 & 0 & 0 & 0 & 0 & 0 & 4 & 0 & 0 & 1 & 0 & 1 \\
        4 & 1 & 4 & 2 & 3 & 3 & 0 & 2 & 0 & 0 & 0 & 2 & 3 & 0 & 1 & 0 & 1 \\
        5 & 1 & 4 & 1 & 2 & 4 & 0 & 1 & 0 & 0 & 0 & 2 & 1 & 0 & 0 & 2 & 1 \\
        6 & 1 & 1 & 3 & 0 & 0 & 2 & 0 & 0 & 0 & 3 & 0 & 0 & 0 & 0 & 0 & 0 \\
        \bottomrule
      \end{tabular}%
    }
  \end{subfigure}
  \caption{An illustration of graphlet transform: the graph $G=(V,E)$
    to the left, with $|V|=6$ and $|E|=9$, is transformed to the
    net frequency vector field $\{ f(v), v \in V \} $ placed in the
    bottom table to the right, with respect to dictionary $\Sigma_{16}$.
    The net frequencies are converted from the raw
    frequencies vector field $\{ \hat{f}(v), v\in V \}$ in the top
    table.
    {\bf Observations.} The transform quantifies and recognizes
    topological properties of graph $G$. 
    The vertices in the same orbit have the same frequency
    vectors, $ f(v_2) = f(v_5) $; $G$ has $5$ triangles, $\mbox{sum}(d_4)/3 = 5$; 
    $G$ is free of 4-cycles, $d_{12} =0$; and free of claws, $d_7 = d_8 = 0$.  }
  \label{fig:briki}
\end{figure*}

\begin{table}[]
  \centering
  \caption{Formulas for fast calculation of raw graphlet frequencies
    on the vertices of a graph $G$ with adjacency matrix $A$, with
    respect to graphlet dictionary $\Sigma_{16}$ as shown in
    \Cref{fig:graphlets}. The auxiliary vectors and matrices are
    specified in
    \Cref{sec:transform-with-Sigma5,sec:transform-with-Sigma15}.  The
    sparse/rectified difference $\max\{b\!-\! a, 0\}$ between two
    vectors $a$ and $b$ is denoted simply as $b\!-\! a$.  }
  \label{tab:summary-formulas}
  \resizebox{\linewidth}{!}{%
    \begin{tabular}{@{}lll}
      \toprule
      \toprule
      \multicolumn{1}{l}{\hspace{-1em} $\Sigma_{16}\!$ }
      & \multicolumn{1}{l}{Graphlet, incidence node} 
      & \multicolumn{1}{l}{Formula in vector expression}
      \\ 
      \toprule
      $\sigma_{0} $ 
      & singleton 
      & $\hat{d}_{0} = e$
      \\ \midrule
      $\sigma_{1} $
      & 1-path, at an end 
      & $\hat{d}_{1} = p_1 $    %
      \\ \midrule         
      $\sigma_{2}$
      &  2-path, at an end 
      & $ \hat{d}_{2} = p_2 $   %
      \\ 
      $\sigma_{3}$
      & bi-fork, at the root 
      & $\hat{d}_{3} = p_1 \odot  (p_1 - 1) / 2$
      \\                 
      $\sigma_{4}$
      & 3-clique, at any node 
      & $ \hat{d}_{4} = c_3 $    %
      \\ \midrule
      $\sigma_{5}$
      & 3-path, at an end 
      & $ \hat{d}_{5} = p_3$  %
      \\
      $\sigma_{6}$
      & 3-path, at an interior node
      &  $\hat{d}_{6} = p_2 \odot ( p_{1} - 1) - 2\, c_3 $ %
      \\  [0.5em] \hdashline\noalign{\vskip 0.5ex}
      $\sigma_{7}$
      & claw,  at a leaf
      & $ \hat{d}_{7} = A \, \big( (p_1 - 1) \odot (p_1 - 2) \big) / 2$
      \\
      $\sigma_{8}$
      & claw,  at the root 
      & $ \hat{d}_{8} = p_1 \odot (p_1 - 1) \odot (p_1 - 2) / 6$
      \\
      $\sigma_{9}$
      & paw, at the handle tip 
      & $ \hat{d}_{9} = A \, c_{3} - 2\,  c_{3} $
      \\ 
      $\sigma_{10}$
      & paw, at a base node 
      &  $ \hat{d}_{10} = C_{3} \, (p_{1} - 2) $ %
      \\
      $\sigma_{11}$
      & paw, at the center
      &  $ \hat{d}_{11} = (p_1 - 2) \odot c_3 $
      \\  [0.5em] \hdashline\noalign{\vskip 0.5ex}
      $\sigma_{12}$
      & 4-cycle, at any node 
      & $ \hat{d}_{12} = c_4 $  %
      \\  [0.5em] \hdashline\noalign{\vskip 0.5ex}
      $\sigma_{13}$
      & diamond, at an off-cord node 
      & $\hat{d}_{13}  = D_{4,c} \,  e/2  $   %
      \\ 
      $\sigma_{14}$
      & diamond, at an on-cord node
      & $ \hat{d}_{14} = D_{4,3} \, e / 2 $ %
      \\
      $\sigma_{15}$
      & 4-clique, at any node  
      & $ \hat{d}_{15}   =  T \,  e/6 $ %
      \\
      \bottomrule
      \bottomrule
    \end{tabular}%
  }
\end{table}

\begin{table}[tbh]
  \caption{The matrix $U_{16} $ for conversion from net frequencies to
    raw frequencies, $ U_{16} f = \hat{f} $, associated with
    dictionary $\Sigma_{16}$. The raw-to-net frequency conversion
    $ f = U_{16}^{-1} \hat{f} $ is used in the fast transform.  All
    coefficients of $U_{16}$ are non-negative. A sub-dictionary with
    index set $s$ has the conversion matrix $U_{16}(s,s)$,
    $\{0,1\} \subseteq s \subseteq \{ 0,1, \cdots 15\} $. }
  \label{tab:frequency-conversion-matrix}
  \centering
  \resizebox{\linewidth}{!}{%
    \begin{tabular}{|l|c|c|c|c|c|c|c|c|c|c|c|c|c|c|c|c|}
      \hline
      $\! U_{16}\!$ 
      & $d_{0} $ & $d_{1} $ & $d_{2} $
      & $d_{3} $ & $d_{4} $ 
      & $d_{5} $ & $d_{6} $ & $d_{7} $
      & $d_{8} $ & $d_{9} $ & $d_{10}$
      & $d_{11}$ & $d_{12}$ & $d_{13}$
      & $d_{14}$ & $d_{15}$ \\
      \hline
      $\hat{d}_{0} $  & 1 &   &   &   &   &   &   &   &   &   &   &   &   &   &   &   \\ \hline
      $\hat{d}_{1} $  &   & 1 &   &   &   &   &   &   &   &   &   &   &   &   &   &   \\ \hline
      $\hat{d}_{2} $  &   &   & 1 &   & 2 &   &   &   &   &   &   &   &   &   &   &   \\ \hline
      $\hat{d}_{3} $  &   &   &   & 1 & 1 &   &   &   &   &   &   &   &   &   &   &   \\ \hline
      $\hat{d}_{4} $  &   &   &   &   & 1 &   &   &   &   &   &   &   &   &   &   &   \\ \hline
      $\hat{d}_{5} $  &   &   &   &   &   & 1 &   &   &   & 2 & 1 &   & 2 & 4 & 2 & 6 \\ \hline
      $\hat{d}_{6} $  &   &   &   &   &   &   & 1 &   &   &   & 1 & 2 & 2 & 2 & 4 & 6 \\ \hline
      $\hat{d}_{7} $  &   &   &   &   &   &   &   & 1 &   & 1 & 1 &   &   & 2 & 1 & 3 \\ \hline
      $\hat{d}_{8} $  &   &   &   &   &   &   &   &   & 1 &   &   & 1 &   &   & 1 & 1 \\ \hline
      $\hat{d}_{9} $  &   &   &   &   &   &   &   &   &   & 1 &   &   &   & 2 &   & 3 \\ \hline
      $\hat{d}_{10}$  &   &   &   &   &   &   &   &   &   &   & 1 &   &   & 2 & 2 & 6 \\ \hline
      $\hat{d}_{11}$  &   &   &   &   &   &   &   &   &   &   &   & 1 &   &   & 2 & 3 \\ \hline
      $\hat{d}_{12}$  &   &   &   &   &   &   &   &   &   &   &   &   & 1 & 1 & 1 & 3 \\ \hline
      $\hat{d}_{13}$  &   &   &   &   &   &   &   &   &   &   &   &   &   & 1 &   & 3 \\ \hline
      $\hat{d}_{14}$  &   &   &   &   &   &   &   &   &   &   &   &   &   &   & 1 & 3 \\ \hline
      $\hat{d}_{15}$  &   &   &   &   &   &   &   &   &   &   &   &   &   &   &   & 1 \\ \hline
    \end{tabular}%
  }
\end{table}

\section{High-performance implementation} 
\label{sec:high-performance-graphlet-transform}

We address the high-performance aspect of graphlet transform.  The
fast graphlet transform has the unique property that the formulas
are simple and in ready form to be translated to high-performance
program and implementation.
We highlight three conceptual and operational issues key to
high-performance implementation.

The first is on the use of sparse masks.  We exploit graph sparsity in
every fast formula.  This is to be formally translated into any
implementation specification: every sparse operation is associated
with source mask(s) on input data and target mask on output
data. Particularly, an unweighted adjacency matrix serves as its own
sparsity mask.
Masked operations are supported by {\tt GraphBLAS}, the output
matrix/vector is computed or modified only where the mask elements are
on, not off.
A simple example is the Hadamard product of two
matrices. As the intersection of two source masks, the target mask is
no denser than any of the source masks. Often, a factor matrix is
either the adjacency matrix $A$ itself or as sparse as $A$.
A non-trivial example is the chain of masks with a sequence of sparse
operations.  For example, in calculating the scalar $v^{\rm T}Av$ with
sparse vector $v$ and sparse matrix $A$, as in
\cref{eqn:sigma-diamond-off-cord} or \cref{eq:K4-matrix}, the target
mask for $Av$ is the nonzero pattern of $v$.  With sparse masks, we 
reduce or eliminate unnecessary operations, memory allocation and memory
accesses.

The next two issues are closely coupled: operation scheduling and
computing auxiliary matrices on the fly. The objectives are to
minimize the number of matrix revisits and to minimize the amount of
working space memory.
Operations using the same auxiliary matrix are carried out together
with updates on the output while auxiliary matrix elements are
computed on the fly.  No auxiliary matrix is explicitly stored.

With our inital implementation\cite{floros2020c}, the space complexity
is $4m + 2n|\Sigma|$. On the network LiveJournal~\cite{yang2015a}
with $4\,$M nodes, $35\,$M edges, the execution takes only $1$ minute
with $16$ threads on Intel Xeon E5-2640.  On the Friendster network
with $66\,$M nodes and $1.8\,$B links, the execution time is 2
hours and 34 minutes with a single Xeon processor. Our multi-thread
programming is in Cilk~\cite{blumofe1996}.

\section{The main theorem \& its merits}

\newcounter{Pcounter} 
\setcounter{Pcounter}{0}

By the preceding analysis we have the following theorem.
\begin{theorem}[Fast graphlet transform with $\Sigma_{16}$] 
  \label{thm:main-FGellT}
  Let $\Sigma$ be a graphlet dictionary,
  $\Sigma\subseteq \Sigma_{16}$.  Let $G=(V,E)$ be a sparse graph. The
  fast graphlet transform of $G$, by the formulas in
  \Cref{tab:summary-formulas} and the frequency conversion in
  \Cref{tab:frequency-conversion-matrix}, has the time and space
  complexities bounded from above as follows.
\begin{list}{(\alph{Pcounter})}%
{ 
 \usecounter{Pcounter}
 \setlength{\rightmargin}{0.4cm}
 \setlength{\leftmargin}{0.2cm}
 \setlength{\itemsep}{0.5pt}
}
\item Upper bound on space complexity: $ 4 m + 2 n |\Sigma| $.
\item Upper bounds on time complexity:%
  \vspace{-0.5em} 
  \[
    \begin{array}{lll} 
     \big( \, 10 \, \gamma(1) \, m\, + 3\, n\, ) \, |\Sigma|,
     &    &  \sigma_{15} \notin \Sigma, 
      \\ 
      5\, \big( \, c\, d_{\max}\, m + \gamma(n_c)\, d(n_c) m  + n \big) |\Sigma|, 
      &   & \sigma_{15} \in \Sigma,
    \end{array}
    \] 
    where $c \!<\!  d_{\max}$ is a constant prescribed by graph type of
    interest, $d(j)$ is the degree of node $j$ in the order of
    non-increasing degrees, $d(1) \geq d(j) \!\geq\!  d(j+1)$,
    $\gamma(j) = \min\{d(j), \, 2\, \alpha(G) \}$,
    $\alpha(G)$ is the arboricity of $G$, and
    $n_c$ exists at
    $ \arg {\displaystyle \max_{k}} \big\{ \sum_{(i,j) \in E, i,j\leq
      k} (a_i^{\rm T}a_j)^2 < c\, m / d_{\max} \big\} $.
\label{eq:1}
\end{list} 
\end{theorem}
\vspace*{-1em}
A proof is in Appendix~\ref{apdix:K4-complexity}.
We comment on dictionary capacity and selection criteria.  A larger
dictionary offers an exponentially increased encoding range at only
linearly increased computation cost.
Depending on the object of graphlet encoding, the relationships among
graphlets may be taken into consideration.  When the bi-fork graphlet
is used to encode the betweenness among triangle clusters, the
triangle graphlet must be included\cite{floros2020a}.  For claw-free
graph recognition, the quad-node graphlets with claw subgraphs must be
included.

The fast graphlet transform and complexity analysis establish a few
remarkable records, to our knowledge. Practically, the fast transform
enables broader use of graphlets for large network analysis.
Computationally, we use {\em sparse} matrix formulas to effectively
reduce redundancy among neighborhoods and streamline
computation. Theoretically, the complexities on regular graphs, planar
graphs, degree-bounded and arboricity-bounded graphs are of the same
order as, or even lower than, the best existing complexities, some of
the latter are asymptotic, resorting to matrix size that can hardly be
reached/materialized~\cite{kloks2000,yaveroglu2015,duke1995,chiba1985,alon1997}.
The complexity term $d(n_c)\gamma(n_c)\, m$ with
$\Sigma_{16}$ on {\em general} graphs breaks down the barrier at
$d_{\max}^2m$ on sparse graphs or $d_{\max}^3n$ on dense graphs as
long and widely believed.
Our fast method for exact graphlet transform with $\Sigma_{16}$
suggests also the possibility of new algorithms, faster than the
existing ones, for rapid recognition and location of forbidden or
frequent quad-node induced subgraphs for biological network study or
theoretical graph classification.

{\small \textbf{Acknowledgements.} This work is partially supported by
  grant 5R01EB028324-02 from the National Institute of Health (NIH),
  USA, and EDULLL 34, co-financed by the European Social Fund (ESF)
  2014-2020.  We thank the reviewer who suggested the inclusion of
  experimental timing results and code release in the revised
  manuscript.  We also thank Tiancheng Liu for helpful comments.}

\clearpage

  \bibliographystyle{IEEEtranS}

\phantomsection
\label{sec:references}
\addcontentsline{toc}{section}{References}
  \balance
\bibliography{IEEEabrv,ref}

\appendices

\section{Bounding the complexity of  $K_3$ counting}
\label{apdix:K3-complexity} 

Let $G=( V, E)$ with $ n = |V| $ vertices and $ m = |E| $ edges.
We bound on the complexity for counting triangles in $G$.
Define first the following $n\times m$ matrix, 
\begin{equation}
  \label{eq:Ce3-matrix}
    C^{e}_3 \triangleq [\,  a_i \odot a_j ,  \,  (i,j) \in E \, ] , 
\end{equation}
where
$a_k = A\, e_k$.  This matrix is actually the triangle-listing matrix:
every column $a_i \odot a_j$ is the indicator of all triangle nodes
opposite to the same base edge $(i,j)$. Matrix $C^{e}_3$ is related to the
triangle-counting matrix $C_3$, as defined in
\Cref{lemm:cycle-sparse}, by
\begin{equation}
  C_3(i,j) = e^{\rm T}(a_i \odot a_j) \leq \min\left\{ d(i), d(j) \right\}. 
\end{equation}
\Cref{lemm:triangle-total} is the short version of the following lemma.
\begin{lemma}[Triangle count and counting cost]
  \label{lemm:appendix-triangle-total}
  The total number of triangles in $G$ is $e^{\rm T} C_3\, e/6$.
  Denote by $\mathrm{cost}(C_3\,e)$ the cost for computing
  $C_3\,e$, the vector of triangle counts at
  at all vertices. Then,
\begin{subequations}
  \label{eq:appendix-K3-Ubound}
  \begin{align}
    \label{eq:triangle-count-cost}
    e^{\rm T} C_3 e \leq \mathrm{cost}(C_3e) 
    &  \displaystyle 
     \leq \sum_{(i,j) \in E  }  \min\{ d(i), d(j) \}
    \\`
   \label{eq:triangles-upper-bond}    
    & \displaystyle 
      \leq \min\{ d_{\max}, 2\alpha(G) \}\, m,
  \end{align}
\end{subequations}
where $\alpha(G)$ is the arboricity of graph $G$.
\end{lemma}
Three remarks. First, the lemma relates and bounds the count and
counting cost by the same summation on the right in
\Cref{eq:triangle-count-cost}.  A sixth of the summation is a tight
upper bound on the number of triangles in a graph. It is $1/6$ on a
star graph, i.e., implying correctly that a star graph is free of
triangles.
Second, the upper bound of \Cref{eq:triangles-upper-bond} is based on,
and improves upon, the upper bound $2\alpha(G)$ by Chiba and Nishizeki
(1985)~\cite{chiba1985}.  The improved bound has immediate implications
on particular types of graphs. For regular graphs, the number of
triangles and the complexity are linear in $m$, while $\alpha(G)$ can
be as high as $n/2$.  For planar graphs, $\alpha(G) \leq 3$, or any
arboricity-bounded graphs, the total number of triangles and triangle
counting cost are linear in $m$.
Third, the complexity of triangle frequencies sets the base for
quad-node graphlet frequencies.

\section{Bounding the complexity of $K_4$  counting} 
\label{apdix:K4-complexity} 

The vector of $K_4$ frequencies is $d_{15} = T\, e/6$, with
$ T \triangleq A \odot \left[ q_{ij}^{\mathrm{T}} A q_{ij}\right] $,
where $q_{ij} = a_i \odot a_j$ with $a_j = A\, e_j$. We have
\begin{equation}
\label{eq:K4-Ubound}
\mathrm{cost}( T)   \leq   \sum_{(i,j) \in E  }  (a_i^{\rm T}a_j )^2 . 
\end{equation}
We give an upper bound with an analysis technique using edge
partition.  Denote by $d(j) $ the degree of node $j$ in the order of
non-decreasing degrees, $d(j) \geq d(j+1)$.  Let $n_c$ be a node
index, to be determined.  We partition the vertices into two disjoint
sets: $H = \left\{ x \middle| \, d(x) > d(n_c) \right\} $ and
$ L = V- H$. The vertex partition induces an edge partition
\begin{equation}
\label{eq:edge-partition} 
  \begin{array}{lll}
    E_1 :   &  a_i^{\rm T}a_j  \leq d_{\max}   ,  &  i, j  \in H ,   
    \\
    E_2 :  &   a_{i}^{\rm T}a_j  \leq d(n_c)  , & \mathrm{ otherwise }. 
  \end{array}
\end{equation} 
Then,
\begin{equation}
  \label{eq:K4-E-split}
  \begin{array}{rcl} 
    \mathrm{cost}( T)
    & \leq &
          \displaystyle             
       d^{2}_{\max} |E_1| + d^{2}(n_c) \left( m - |E_1| \right )        
    \end{array} 
\end{equation}

Fix a small constant $c $ a priori.  If $d_{\max} \leq c$, then,
$ \mathrm{cost}( T) \leq c^2 \, m $.  Otherwise, we determine or
locate $n_c$ in the following way. Let
$ | E_1 | \triangleq \left\lceil c\, m / d_{\max} \right\rceil. $ 
Clearly, $|E_1| < m$. Let $n_c$ be the node index at which $|E_1|+1$ is
reached, i.e., $n_c$ is the node location for the desired edge
partition. Then,
\begin{equation}
  \label{eq:K4-Ubound-2}
  \begin{array}{rcl} 
    \mathrm{cost}( T )
    & \leq & 
    \displaystyle 
             c\, d_{\max}\,  m + \gamma(n_c) d(n_c)  \, m, 
    \\
    \gamma (n_c) &  = &
    \displaystyle     
                        \min \left\{  d(n_c) (1 \!-\! c/d_{\max}) , 2\, \alpha (G )  \right\},
  \end{array}
  \end{equation}
  where \Cref{eq:triangles-upper-bond} is applied to the second term. 
  We have proved the upper bound with $\Sigma_{16}$ in \Cref{thm:main-FGellT}.

  The factor $\gamma(n_c)$ in the upper bound accommodates
  the variation with graph types or degree distributions. In particular, 
    \begin{enumerate}
    \item For  regular graph with degree $k$, the bound
      in \Cref{eq:K4-Ubound-2} recovers to be $k^2 \, m$.
    \item For planar graph, $\alpha (G) \leq 3$.  One may set $c=6$ to
      let the first term $c\, d_{\max}\, m$ dominate.
    \item For scale-free or small-world networks, $d(n_c)\gamma(n_c)$
      is smaller than $d_{\max}$ probabilistically.
    \end{enumerate}

\end{document}